\documentclass[letterpaper, 10 pt, conference]{ieeeconf}  

\usepackage{todonotes}

\IEEEoverridecommandlockouts                              

\usepackage{graphics} 
\usepackage{amsmath} 
\usepackage{amssymb}  
\usepackage{algpseudocode,algorithm,algorithmicx}
\usepackage{hyperref}
\usepackage{mathtools}
\usepackage{cite}
\usepackage{soul}
\usepackage{comment}
\usepackage[outdir=./]{epstopdf}

\def\IR{{\rm I\!R}}
\def\IC{{\rm I\hspace{-1.2ex}C}}
\def\ni{n}
\def\no{q}
\def\np{m}
\def \v{\upsilon}
\def\P{\mathcal{P}}
\def\p{p}
\title{\LARGE \bf
Structure-preserving Model Reduction of Parametric Power Networks$^{*}$
\thanks{\textcolor{black}{*This work was supported in parts by National Science Foundation under Grant No. DMS-1923221.}}}

\author{Bita Safaee$^{1}$ and Serkan Gugercin$^{2}$
\thanks{$^{1}$B.~Safaee is with the Department of Mechanical Engineering, Virginia Tech, Blacksburg, VA 24061,
        {\tt\small  bsafaee@vt.edu}}%
\thanks{$^{2}$S.~Gugercin is with the Department of Mathematics
Virginia Tech, Blacksburg, VA 24061,
        {\tt\small gugercin@vt.edu}}%
}

\newtheorem{theorem}{Theorem}[section]
\newtheorem{remark}{Remark}[section]
\newtheorem{corollary}{Corollary}[section]
\newtheorem{proposition}{Proposition}[section]
\newcommand{\RN}[1]{%
  \textup{\uppercase\expandafter{\romannumeral#1}}%
}

\begin{document}

\maketitle
\thispagestyle{empty}
\pagestyle{empty}

\begin{abstract}
We develop a structure-preserving parametric model reduction approach for linearized swing equations where 
parametrization corresponds to variations in operating conditions. 
We employ a global basis approach to develop the parametric reduced model
in which we concatenate the local bases obtained via
 $\mathcal{H}_2$-based interpolatory model reduction.
The residue of the underlying dynamics corresponding to the simple pole at zero varies with the parameters. Therefore, to have  bounded $\mathcal{H}_2$ and $\mathcal{H}_\infty$ errors, the reduced model residue for the pole at zero should match the original one over the entire parameter domain. Our framework achieves this goal by enriching the global basis based on a residue analysis.  
The effectiveness of the proposed method is illustrated through two numerical examples. 

\end{abstract}

\section{INTRODUCTION}
Power networks  are naturally modeled as second-order dynamical systems \cite{Kundur94,nishikawa2015comparative,sauer1998power,ChengKawano2017}. 
In the case of large-scale networks,
monitoring, analysis and control of resulting second-order systems become exceedingly difficult due to unmanageable computational demands. To tackle this predicament, we 
apply model reduction in which the goal is to construct a lower dimensional model that preserves the physically meaningful second-order dynamics and provides a high-fidelity approximation of the input/input behaviour. There is a plethora of model reduction approaches for  second-order dynamical systems, see, e.g., \cite{BaiSu2005}, \cite{BONINFabbender2016},  \cite{BEATTIEGugercin2009}, \cite{SuCraig1991}, \cite{MeyerSrinivasan1996}, \cite{ReisStykel2008},
\cite{ChaGVV05},
for model reduction of general second-order systems, and 
see, e.g., \cite{Ishizaki2015}, \cite{ChengKawano2017}, \cite{ChengScherpen2016},
\cite{mlinaric2020thesis}, \cite{YuCheng2019}
with a focus on
network dynamics. 

In this paper, we focus on parametrically varying power networks where the parameter variations correspond to different operation conditions. 
This leads to the parametric model reduction (PMOR) framework \cite{SerkanSurvey,BCOW17,QuarteroniMN16,hesthaven2016certified}.
The goal of PMOR is to find a parametric reduced model that can approximate the original model with acceptable fidelity over a wide range of 
parameters. PMOR eliminates the need for performing a separate reduction at each parameter value (operating condition) and 
 therefore  plays an important rule in control, design, optimization and uncertainty quantification. 

To form our parametric reduced-order structure-preserving (second-order) power network model we employ a global basis approach where the model reduction basis is constructed by concatenation of local bases for selected parameter samples. We obtain the local bases 
using second-order interpolatory $\mathcal{H}_2$-optimal methods \cite{wyatt2012issues,TomljBeattieGugercin18}. Since the full-order dynamics has a pole at zero with a parametrically varying residue, the parametric reduced model needs to retain this residue in order to have bounded $\mathcal{H}_2$ and  $\mathcal{H}_{\infty}$ error norms for whole parameter domain. Based on a detailed residue analysis, we establish the subspace conditions on the model reduction basis to guarantee this property and explain the algorithmic implications.

{The remainder of this paper is organized as follows: Section \ref{sec:intro}  presents nonlinear model of the swing equations as well as its corresponding non-parametric and parametric second-order linear approximations. In Section \ref{sec:modred}, we describe the parametric reduction method via interpolatory model reduction bases. Section \ref{sec:residue} presents our main theoretical results 
for subspace conditions to guarantee parametric residue-matching together with computational details. Section \ref{sec:num} illustrates the feasibility of our approach via numerical examples followed by conclusions 
in Section \ref{sec:conc}.}
\section{Network swing model}   \label{sec:intro}
A power network can be represented by a connected graph $\mathcal{G} = (\mathcal{V},\mathcal{E})$ with buses as nodes $\mathcal{V} = \{ 1,\dots,n \} $ and transmission lines as edges $ \mathcal{E} \subseteq \mathcal{V} \times \mathcal{V}$. Generally, a bus can host different combinations of generators and loads,  or it may even be a simple junction node.
Assume that each bus hosts a generator. We can model the active power ${P_{ij}}$ flowing from bus (node) {$i$} to bus {$j$} along the transmission line {$(i,j) \in \mathcal{E}$} as 
\begin{equation} \label{P_nm}
P_{ij} = \frac{E_i E_j}{\chi_{ij}}\sin (\delta_i - \delta_j),
\end{equation}
where $\delta_i$ is the phase angle, {$E_i$} is the peak voltage magnitude, {and $\chi_{ij} > 0$  is the line reactance}. This model ignores the line resistances. The swing equation for a single generator {$i$} results from  Newton's second law and is given by
\begin{align}
 \label{single swing}
M_{i}\ddot{\delta_i}+D_{i}\dot{\delta_i} = P^{mech}_{i} - P^{elec}_{i}, \  i \in \{1,\dots,n\},
\end{align}
where $M_i > 0$ is the rotor moment
of inertia, $D_i > 0$ is a damping constant, and $P^{mech}_{i}$ and $P^{elec}_{i}$
are the input mechanical power and output electrical power for the $i^{th}$ generator, respectively. 
Combing (\ref{P_nm}) and (\ref{single swing}) leads to the swing equations of an electric power grid 
\cite{nishikawa2015comparative,sauer1998power,Bergen1981}
\begin{align} \label{swing equation}
 M_{i }\ddot{\delta_i}&+D_{i}\dot{\delta_i} + \sum_{j \in \mathcal{V}_i} \frac{E_i E_j}{\chi_{ij}}\sin (\delta_i - \delta_j)\\
 &= P^{mech}_{i} - P^{load}_{i} = P_i^{net} \ , \ \ \ \ \ \ \forall i \in \mathcal{V}, \notag
\end{align}
where the set $\mathcal{V}_i \in \mathcal{V}$ refers to those buses {connected to} bus $i$ in $\mathcal{G}$, $P^{load}$ corresponds to the portion of the electric power consumed at bus 
$i$ and $P_i^{net}$ is the net power input at bus 
$i$. 

Assuming small angle differences ($\delta_i - \delta_j \simeq 0)$ and unity voltage magnitudes ({$E_i= 1$}), we can rewrite (\ref{P_nm}) as
\begin{equation}
 P_{ij} \simeq b_{ij}(\delta_{i} - \delta_{j}),  
\end{equation}
where $b_{ij} = \frac{1}{\chi_{ij}}$ is the suseptance between the nodes $(i,j)\in \mathcal{E} $. 
{Define 
$\delta = [\delta_1,\delta_2,\ldots,\delta_n]^T \in {\rm I\!R}^{n}$.} 
Then, the original dynamics in
 (\ref{swing equation}) can be linearized as
\begin{equation} \label{linear model}
    \Sigma :=
        \begin{cases}
$$M\ddot{\delta}(t)+D\dot{\delta}(t)+ L\delta(t) = B u(t),$$\\
$$ y(t) = C\delta (t)$$ ,
    \end{cases}
\end{equation}
where
$M = \mathsf{diag}(M_1,M_2,\ldots,M_n)\in {\rm I\!R}^{n \times n}$ and $D = \mathsf{diag}(D_1,D_2,\ldots,D_n)\in {\rm I\!R}^{n \times n}$ are the diagonal matrices of inertia and damping coefficients, and {$L \in {\rm I\!R}^{n \times n}$ is the susceptance Laplacian matrix} ($L = L^T \geq 0$) whose {$(i,j)$th} entry is given by 
\begin{equation} \label{Lnm}
   [ L]_{i,j} := 
    \begin{cases}
    $$-b_{ij}$$ , & \text{if $(i,j) \in \mathcal{E}$,}\\
    $$ \sum_{(i,j) \in \mathcal{E}} b_{ij}$$ , & \text{if $ j = i$,}\\
    0, & \text{otherwise.}
    \end{cases}
\end{equation}
 {Moreover
 $u = [ P_1^{net}~\ldots~P_n^{net}]^T \in \IR^n$, 
 $B \in {\rm I\!R}^{n \times \ni}$ is the identity matrix}, and {$C \in {\rm I\!R}^{\no \times n}$ yields the output of the system}. 
{By defining the new state variable 
$x = [\delta^T~\dot{\delta}^T]^T \in {\rm I\!R}^{2n}$}, one can equivalently represent the second-order dynamic (\ref{linear model}) in its first-order form 
\begin{align} 
\dot{x} = \mathcal{A}x + \mathcal{B}u,~~~y(t) = \mathcal{C}
x   \label{firstordernp}
\end{align} with 
$\mathcal{A} = \begin{bmatrix}
  0 & I\\ -M^{-1}L & -M^{-1}D \end{bmatrix} \in 
 {\rm I\!R}^{(2n) \times (2n)}$, 
 $\mathcal{B} = \begin{bmatrix}
       0\\M^{-1}B
  \end{bmatrix}\in 
 {\rm I\!R}^{(2n) \times \ni}$,
 and $\mathcal{C} = \begin{bmatrix} 
        C & 0
  \end{bmatrix}\in 
 {\rm I\!R}^{\no \times (2n)}$, where $I \in {\rm I\!R}^{n \times n}$ is the identity matrix.
Due to the simple zero eigenvalue of $L$, 
 $\mathcal{A}$ has one eigenvalue at zero and $2n-1$ eigenvalues in the left-half plane. {Thus  (\ref{linear model}) is {a stable} dynamical system, not asymptotically stable \cite{ChengKawano2017}.}
\subsection{Linearized parametric model}
In practice, matrix $L$ is not constant due to  variations, for example, in  peak voltage magnitudes {$E_i$}.  {Therefore, to allow variations, we will view
$E_i$ as a parameter that can vary and write it simply as $\p_i$}. This leads to the parametric power network model that appears as
\begin{align} 
 M_{i }\ddot{\delta_i}(t;\p)+&D_{i}\dot{\delta_i}(t;\p) 
 + \sum_{j \in \nu_i} \frac{ \textcolor{black}{p_i} \textcolor{black}{ p_j}}{\chi_{ij}}\sin \notag (\delta_i(t;\p) - \delta_j(t;\p)) \\ 
& = {P^{net}_{i}}  \ \ ,\ \ \ \ \forall i \in \mathcal{V} \label{swing equation p}
\end{align}
with the corresponding linear model 
\begin{equation} \label{param swing}
\begin{cases}
$$ M\ddot{\delta}(t;\p)+D\dot{\delta}(t;p)+ L(\p)\delta(t;\p) = {B} u(t),$$\\
$$ y(t;\p) = {C}\delta (t;\p), $$
\end{cases}
\end{equation}
where  $\p = \begin{bmatrix}
 \p_1 & \p_2 & \dots & \p_n     
\end{bmatrix}^T \in {\Omega \subseteq \IR^n}$ is the parameter vector, 
the  matrix $L(\p)$ will now vary with  $\p$, and allows for variation in operating conditions. The parametric  matrix $L(\p)$ can be written as
\begin{equation} \label{L_p}
  L(\p) = \P L\P, 
\end{equation}
where {$\P=\mathsf{diag}(\p) = \mathsf{diag}(\p_1,\ldots,\p_n) \in {\rm I\!R}^{n \times n}$} is diagonal  and {$L$ is as defined in} (\ref{Lnm}). {Note that $p_i=1$ for $i=1,\ldots,n$ recovers the non-parametric problem. We will allow $p_i$'s vary around this nominal value, i.e., $p_i \in (1-\alpha,1+\alpha)$ where  $0<\alpha<1$; thus $\P$ stays invertible for every $p \in \Omega$. Choosing, e.g., $\alpha =0.15$, corresponds to allowing a $15\%$variation in peak voltage magnitudues.} 
\section{Structure-preserving parametric reduced models for linearized swing equations} \label{sec:modred}

 We seek to develop a reduction framework such that not only it preserves the structure, but also the parametric reduced model serves with acceptable accuracy as a surrogate model over diverse operating conditions. 
Since it is crucial that the reduced model  preserves the physically-meaningful second-order structure, instead of transferring the second-order dynamics to the first-order form, as in (\ref{firstordernp}), and applying model reduction there, we will directly reduce the second-order dynamics (\ref{param swing}). 
 In other words, our goal is to find a reduced parametric system 
\begin{align} 
\begin{array}{l}
 M_r\ddot{\delta_r}(t;\p)+D_r\dot{\delta_r}(t;\p)+ L_{{r}}(\p)\delta_r(t;\p) = B_r u(t) \\
 ~\phantom{\quad} y_r(t;\p) = C_r \delta_r (t;\p), 
 \end{array} 
 \label{param swing reduced}
\end{align}
where $M_r$, $L_r(\p)$, $D_r \in {\rm I\!R}^{r \times r}$,  $B \in {\rm I\!R}^{r \times \ni}$ and $C \in {\rm I\!R}^{\no \times r}$ with $r \ll n $ such that the $y_r(t;\p)\approx y(t;\p)$ for a wide range of inputs $u(t)$ over the parameter range of interest. 

Since $M$ and $D$ are symmetric positive definite, and $L(\p)$ is symmetric positive semi-definite, one should preserve these structures in the reduced  model. 
We  achieve this  using Galerkin projection: construct a model reduction basis $V \in {\rm I\!R}^{n \times r}$  and 
the reduced-order matrices in (\ref{param swing reduced}) using
\begin{align} \label{parametric reduced matrices}
& M_r = V^T M V, \ D_r = V^T D V, \ L_{{r}}(\p) = V^T L(\p) V, \\
& B_r = V^T B ,\ \mbox{and}\  C_r = C V.  \notag
\end{align}
Accuracy of the structure-preserving reduced model (\ref{param swing reduced}) with the form (\ref{parametric reduced matrices}) clearly depends on the choice of $V$. We describe this choice next.

\subsection{Interpolatory model reduction bases} \label{intmodred}
There are numerous ways to choose the model reduction basis $V$ for reducing parametric  dynamical systems; see, for example, \cite{SerkanSurvey,BCOW17,QuarteroniMN16,AntBG20,hesthaven2016certified} and the references therein. For the parametric structured second-order dynamical system  (\ref{param swing}), we will employ the structure-preserving parametric interpolatory model reduction framework from \cite{Ant2010imr}, which extended the interpolatory model reduction framework for parametric systems \cite{BaurBeattieBennerGugercin2011} to the structured setting. For recent extensions of structured interpolatory model reduction to special classes of nonlinear systems, see \cite{morBenGW20a,morBenGW20b}. 

Transfer functions of the full-order parametric model (\ref{param swing}) and reduced one  (\ref{param swing reduced}) are, respectively, given by
\begin{align}
 H(s,\p) &= C (s^2 M + s D + L(\p))^{-1} B,\quad\mbox{and} \\
 H_r(s,\p) &= C_r (s^2 M_r + s D_r + L_{{r}}(\p))^{-1} B_r. 
\end{align}
Note that both  $H(s,\p)$ and $ H_r(s,\p)$ are {$\no \times \ni$} matrix-valued rational functions in $s$.  The goal, in  parametric interpolatory model reduction, is to choose $V$ such that  $H_r(s,\p)$ interpolates $ H(s,\p)$ at selected  points in the frequency $s$ and parameter $\p$. Since $H(s)$ is matrix-valued, one enforces interpolation only along the selected directions: Let $\p^{(i)}$ be a parameter point of interest. And let $\{\sigma_1^{(i)},\ldots,\sigma_{r_i}^{(i)}\} \in \IC$ be the frequency interpolation points with the corresponding tangent directions  $\{b_1^{(i)},\ldots,b_{r_i}^{(i)}\} \in \IC^\ni$ for the parameter sample $\p^{(i)}$. 
Assume we have $\np$ parameter samples $\{\p^{(1)},\ldots,\p^{(\np)}\}$. Then, the goal is to construct $V$ such that
\begin{align} \label{intcond}
 H(\sigma_j^{(i)},\p^{(i)})b_j^{(i)} = H_r(\sigma_j^{(i)},\p^{(i)})b_j^{(i)}
\end{align}
for $j=1,2,\ldots,r_i$ and $i=1,2,\ldots, \np$. 
 
 Define {$\mathcal{K}(s,\p) = s^2 M + s D + L(\p)$}. For $i=1,2,\ldots,\np$, construct  the \emph{local} interpolation basis $V^{(i)}  \in
 \IC^{n\times r_i}$
 corresponding to the parameter sample $\p^{(i)}$ using
 $$
 V^{(i)}= [\mathcal{K}(\sigma_1^{(i)},\p^{(i)})^{-1}Bb_1^{(i)},\ldots,\mathcal{K}(\sigma_{r_i}^{(i)},\p^{(i)})^{-1}Bb_{r_i}^{(i)}] 
 $$
and concatenate the local bases to construct the \emph{global} basis:
\begin{align} \label{parametric reduction bases}
& V  = \mathsf{orth}\left(\begin{bmatrix}
     V^{(1)} & V^{(2)} & \dots V^{(\np)}
\end{bmatrix} \right) \in {\rm I\!R}^{n \times r}, 
\end{align}
where ``$\mathsf{orth}$" refers to an orthogonal basis so that $V^T V = I_r.$ Realness of $V$ is guaranteed by choosing the interpolation points and tangent directions in conjugate pairs. Then, the reduced model (\ref{param swing reduced}) obtained as in (\ref{parametric reduced matrices}) using $V$ from 
(\ref{parametric reduction bases}) satisfies the interpolation conditions (\ref{intcond}); see \cite{Ant2010imr,AntBG20}.

Quality of the reduced model will depend on  the choice of interpolation points and tangent directions. In this paper, we choose them, and thus  the local bases $V^{(i)}$, using interpolatory optimal $\mathcal{H}_2$ model reduction. In other words, for every $\p^{(i)}$, we construct the local basis $V^{(i)}$ to minimize/reduce the $\mathcal{H}_2$-distance 
\begin{align}
\| H(\cdot,\p^{(i)}) - H_r(\cdot,\p^{(i)}) \| _{\mathcal{H}_2} &= \\
\Big(\frac{1}{2\pi}\int_{-\infty}^{\infty}   \|H(\imath\omega,\p^{(i)})&-H_r(\imath\omega,\p^{(i)})\|_F^2 d\omega\Big)^{\frac{1}{2}}, \nonumber
\end{align}
where $\imath^2 = -1$ and $\| \cdot \|_F$ denotes the Frobenius  norm. Optimal $\mathcal{H}_2$ model reduction is a heavily studied topic, In the case of unstructured linear dynamical systems, i.e.,
$H_r(s) = C_r(s I_r - A_r)^{-1}B_r$, the optimal reduced model in the $\mathcal{H}_2$-norm is a bitangential Hermite interpolant to $H(s)$ at the mirror images of the reduced poles \cite{gugercin2008hmr,Ant2010imr}. 
The Iterative Rational Krylov Algorithm (IRKA) \cite{gugercin2008hmr} and it variants, e.g., \cite{BeaG12,HokMag2018,xu2010optimal}, have been successfully applied  
in this setting to construct optimal interpolation points and directions. Since we require the reduced-model to have the second-order form, we employ the structured version of IRKA, namely  the Second Order IRKA (SOR-IRKA) \cite{wyatt2012issues,TomljBeattieGugercin18} to construct the local bases $V^{(i)}$. SOR-IRKA
produces a reduced-model that satisfies only a subset of optimal interpolation conditions at the cost of preserving structure. Since the underlying system has a pole at zero in our case, we will modify SOR-IRKA further. This will be explained in detail in Section \ref{sec:alg}. For  other  work on $\mathcal{H}_2$-based model reduction of second-order systems, see, e.g., \cite{beattie2014h2,mlinaric2020thesis,YuCheng2019}.

\begin{remark} 
As opposed to developing locally optimal $\mathcal{H}_2$ model reduction bases $V^{(i)}$ and concatenating them to construct the global basis $V$, 
following \cite{BaurBeattieBennerGugercin2011} one could 
introduce a composite error measure ($\mathcal{L}_2$ error in the
parameter space  and $\mathcal{H}_2$ error in the frequency domain).
Then, one can try to construct $V$
directly to minimize this composite measure. We refer the reader to \cite{BaurBeattieBennerGugercin2011} and more recent works \cite{morHunMS18,grimm2018parametric}  
in this direction
for the unstructured setting. 
\end{remark}
\section{Matching the parametric residue corresponding to the pole at zero}  \label{sec:residue}
Since $L(\p) = \P L \P$ and $L \mathbf{1}= 0$ where $\mathbf{1} \in {\rm I\!R}^{n \times 1}$ is the vector of ones, we obtain
$ L(\p)\P^{-1} \mathbf{1} = \P L \mathbf{1}= 0$.
 Therefore, for every $\p \in \Omega$, $L(\p)$ has a simple zero eigenvalue
 with the eigenvector $\v = \P^{-1} \mathbf{1}$,
and consequently $H(s,\p)$ has a simple pole at zero for every $\p$. This means that $H(s,\p)$ is not
an $\mathcal{H}_2$-function. However, we can still perform
an $\mathcal{H}_2$-based model reduction 
on $H(s,\p)$ as long as we guarantee that the error system, i.e., 
$H(s,\p)-H_r(s,\p)$, stays an $\mathcal{H}_2$-function for every $p$.
{This issue has been studied in the \emph{non-parametric} case.  \cite{ChengKawano2017} achieves a bounded $\mathcal{H}_2$ error norm in model reduction of second order networks where the Galerkin projection is obtained via clustering  techniques. In a more recent work, \cite{YuCheng2019} splits a non-parametric second order network with proportional damping into an asymptotically stable system and an average subsystem containing the zero eigenvalue. Then, the asymptotically stable system is reduced via interpolatory techniques and then re-combined with the average system leads to a reduced model with bounded  (and small) $\mathcal{H}_2$ error. We also refer the reader to, e.g., \cite{jongsma2018model,mlinaric2015efficient,monshizadeh2014projection,mlinaric2020thesis}
for the first-order dynamics case.}

{In reducing the \emph{parametric} second-order model (\ref{param swing}), we need to enforce that $H_r(s,p)$ retains the zero eigenvalue and its \emph{parametric} residue \emph{for every} $\p \in \Omega$ so that the error stays bounded over the whole domain. Next, we establish the subspace conditions on the model reduction basis $V$ to achieve this goal.
}

\subsection{Subspace conditions for matching the parametric residue}
\label{sec:subspace}
For a given a parameter, the next result establishes the conditions on $V$ to match the residue at zero.
{\begin{theorem} \label{Residue}
Given the parametric full-order model 
 (\ref{param swing}),
 let the parametric reduced model (\ref{param swing reduced}) be obtained as in (\ref{parametric reduced matrices}). 
 Let $\hat{\p} \in \Omega$ be a parameter of interest. Define
 $\hat{\P} = \mathsf{diag}(\hat{\p})$
 and 
$ \hat\v = \hat{\P}^{-1}\mathbf{1} $.
 Then for $\hat{\p} \in \Omega$, the reduced model $H_r(s,\hat\p)$ retains the simple pole of $H(s,\hat \p)$ at zero and its corresponding parameter-dependent residue if $\hat\v \in \mathsf{span}(V)$.
\end{theorem}}
\begin{proof}
First, we show that $L_r(\hat \p)$ has a simple zero eigenvalue. Using
$\hat\v \in \mathsf{span}(V)$,
write ${V}$ as
  ${V} =  \begin{bmatrix} 
          V_1 & \hat\v
    \end{bmatrix}$ 
where ${V}_1 \in {\rm I\!R}^{n \times ({r-1})} $ and 
$\hat\v \notin \mathsf{span}({V}_1)$. Then,  using the fact $L(\hat\p)\hat\v = 0$, we obtain
\begin{align} \label{Lrp_hat}
  L_r(\hat\p) = {V}^T L(\hat\p) {V}  = 
     \begin{bmatrix}
           V_1^TL(\hat\p)V_1 & 0 \\ 
           0 & 0
     \end{bmatrix}.
          \end{align}
Since $\hat\v \notin span({V}_1)$, $L_r(\hat\p)$ 
has only one simple zero eigenvalue. Moreover, since $M$ and $D$ are positive definite and  model reduction is performed via a Galerkin projection as in (\ref{parametric reduced matrices}), all the other poles
of $H_r(s,\hat\p)$ have negative real parts except for this simple pole at zero.

Now we need to show that the parametrically varying residues of $H(s,\hat\p)$ and $H_r(s,\hat\p)$ corresponding to 
the pole at zero match. To find the residue of $H(s,\hat\p)$, we follow an analysis inspired by \cite{ChengKawano2017}. Transform the second-order dynamic (\ref{param swing}) to its equivalent first-order form 
\begin{align}
    \dot{x}(t;\p) = \mathcal{A}(p)x(t;\p) + \mathcal{B}u(t),~~~~
   y(t;\p) = \mathcal{C}x(t;\p), \notag
\end{align}
\vspace*{-3ex}
\begin{align}  \notag
  & \mbox{where}~\mathcal{A}(\p) = \begin{bmatrix}
  0 & I\\ -M^{-1}L(\p) & -M^{-1}D \end{bmatrix},\ 
  \mathcal{B} = \begin{bmatrix}
       0\\M^{-1}B
  \end{bmatrix}, \\
  & \mbox{and}~~~ \mathcal{C} = \begin{bmatrix} 
        C & 0
  \end{bmatrix}. \label{first order parametric matrices}
\end{align}
Let $\mathcal{A}(\p) $ have {the Jordan decomposition}
\begin{align} \label{Ajordan}
  \mathcal{A}(\p) = Q \Lambda Q^{-1} = 
  \begin{bmatrix}
      q_1 & {Q_2}
  \end{bmatrix}
  \begin{bmatrix}
  0 & \\  & \bar{\Lambda} \end{bmatrix} 
    \begin{bmatrix}
      \Tilde{q}_1^T \\ \tilde{Q}_2^T
  \end{bmatrix},
\end{align}
where the Jordan block $\bar{\Lambda} \in \IC^{(2n-1)\times (2n-1)}$ contains the eigenvalues with negative real parts, and $q_1$ $\in {\rm I\!R}^{2n}$ and $\tilde{q}_1$ $\in {\rm I\!R}^{2n}$ are, respectively, the right and left eigenvectors corresponding to zero eigenvalue such that
\begin{align} \label{q1}
    \mathcal{A}^{T}(\p) \tilde{q}_1 = 0,\ \mathcal{A}(\p) q_1 = 0, \ {\tilde{q}_1^T q_1 = 1.}
\end{align}
We note that this decomposition is parameter dependent but to simplify the notation, we write, e.g., $Q$ instead of $Q(\p)$.
At $\p=\hat\p$,
 using $L(\hat\p) \hat \v = 0$, and (\ref{Ajordan}) and (\ref{q1}), we obtain
\begin{align} \label{u_1 v_1}
q_1 = 
    \begin{bmatrix}
        \mathbf{\hat\v} \\ 0
    \end{bmatrix} ~~\mbox{and}~~ \tilde{q}_1= \frac{1}{\alpha_D}
    \begin{bmatrix}
        \hat D \hat \v \\  M \hat \v
    \end{bmatrix},
\end{align}
where $\alpha_D = \hat\v^T D \hat\v$. 
Using (\ref{Ajordan}), we write
 \begin{align} 
 H(s,\hat\p) &= \mathcal{C} (s I - \mathcal{A}(\p))^{-1} \mathcal{B} = \mathcal{C}Q (sI - \Lambda)^{-1}Q^{-1}\mathcal{B} \notag \\
 & = \frac{(\mathcal{C}q_1)(\tilde{q}_1^T\mathcal{B})}{s} +
 \mathcal{C}Q_2 (s I - \bar{\Lambda})^{-1} \tilde{Q}_2\mathcal{B}. \label{H_Decomposition}
     \end{align}
Thus,
$\phi_0 = (\mathcal{C}q_1)(\tilde{q}_1^T\mathcal{B})$ is the residue of $H(s,\hat\p)$ for the pole at zero. Then, substituting $q_1$ and $\tilde{q}_1$
from (\ref{u_1 v_1}), and  $\mathcal{C}$ and $\mathcal{B}$ from (\ref{first order parametric matrices})  into $\phi_0 = (\mathcal{C}q_1)(\tilde{q}_1^T\mathcal{B})$ yields
\begin{align}
    \phi_0 
    = \mathcal{C} \alpha_D^{-1} \begin{bmatrix} 
           \hat\v \hat\v^T D & \hat\v \hat\v^T M \\ 
          0 & 0
    \end{bmatrix} \mathcal{B} 
     = \alpha_D^{-1} C \hat\v \hat\v^T B.  \label{phi0}
\end{align}
Similarly, the residue of the reduced system $H_r(s,\hat \p)$  corresponding to the pole at zero is obtained as
\begin{align} \label{phi0red}
    \phi_{0_{r}} = \alpha_{D_{r}}^{-1} C {V} {V}^{T} \hat\v \hat\v^{T} {V} {V}^{T} B,
\end{align}
where $\alpha_{D_{r}}  = \hat\v^{T} {V}  D_r {V}^{T} \hat\v$.

Since ${V} {V}^{T}$ is an orthogonal projector, if $\hat\v \in span({V}) $, we have $V {V}^{T}\hat\v = \hat\v$,
\begin{align}
   \alpha_{D_{r}}  = \hat\v^{T} {V}  D_r {V}^{T} \hat\v=\hat\v^{T} {V} {V}^{T} D {V} {V}^{T} \hat\v
   = \hat\v^{T} D  \hat\v = \alpha_D, \notag
\end{align}
and thus
     $\phi_{0_{r}} = \alpha_{D_{r}}^{-1} C {V} {V}^{T} \hat\v \hat\v^{T} {V} {V}^{T} B = \phi_0$.

\end{proof}
Theorem \ref{Residue} establishes that if $\hat \v  = \hat{\P}^{-1} \mathbf{1} \in \mathsf{span}(V)$, for that parameter value $\hat{\p}$, the residues of $H(s,\hat{\p})$ and $H_r(s,\hat{\p})$ match {for the pole at $s=0$}. This means that 
\begin{align*} \notag 
    & H(s, \hat{\p})-H_r(s, \hat{\p})  \\
    & = \mathcal{C} (s I - \mathcal{A}(\hat{\p}))^{-1} \mathcal{B} - \mathcal{C}_r (s I - \mathcal{A}_r(\hat{\p}))^{-1} \mathcal{B}_r \notag \\
    &= \frac{\phi_0}{s}+ H_a(s,\hat{\p}) -\left( \frac{\phi_{r_{0}}}{s}+ H_{a_{r}}(s,\hat{\p}) \right)\\
    &= H_a(s,\hat{\p})  -H_{a_{r}}(s,\hat{\p}),\notag
\end{align*}
where $H_a(s,\hat{\p}) = \mathcal{C}Q_2 (s I - \bar{\Lambda})^{-1} \tilde{Q}_2\mathcal{B}  $ as in (\ref{H_Decomposition}) and $H_{a_{r}}(s,\hat{\p}) =  \mathcal{C}_r Q_{2_{r}} (s I - \bar{\Lambda}_r)^{-1} \tilde{Q}_{2_{r}}\mathcal{B}_r$ are asymptotically stable. Therefore, the error system is asymptotically stable at $\hat{\p}$. We write this result as a corollary.
\begin{corollary} \label{cor:error}
Assume the set-up of Theorem \ref{Residue}. Then,
the error system  $H(s,\hat{\p})-H_r(s,\hat{\p})$ is asymptotically stable, 
and  has bounded $\mathcal{H}_2$ and $\mathcal{H}_\infty$ norms.
\end{corollary}

\subsection{Algorithmic Implications} \label{sec:alg}
Theorem \ref{Residue} and Corollary \ref{cor:error} hint at how to construct $V$ so that the error system is asymptotically stable at a parameter value of interest. As stated in Section \ref{intmodred}, for the parameter samples $\p^{(i)}$ for $i=1,\ldots,m$, we will construct the local bases $V^{(i)}$ via SOR-IRKA to have local $\mathcal{H}_2$ optimality. However, we will modify SOR-IRKA by taking into consideration that $H(s,\p)$ has a pole at zero for every $\p$, i.e.,
$H(s,p)$ is not an $\mathcal{H}_2$ function. 
{SOR-IRKA} is an iterative algorithm that corrects the interpolation points in every step. Due to the pole at zero, SOR-IRKA will drive one of the interpolation points to zero as it should so that
the pole and residue at zero are matched. This will require computing the vector $\mathcal{K}(0,\p^{(i)})^{-1}Bb_0^{(i)}$.  However, due to the pole at zero, $\mathcal{K}(0,\p^{(i)})$ is not invertible. Therefore, inspired by Theorem \ref{Residue}, in SOR-IRKA, we will  replace this vector with the zero eigenvector of $L(\p^{(i)})$ and thus the span of $V^{(i)}$ will contain this eigenvector. Hence, 
once the global basis  $V$ is constructed as in (\ref{parametric reduction bases}), Theorem \ref{Residue} will guarantee that the error system $H(s,\p)-H_r(s,\p)$ is asymptotically stable for the \emph{sampled parameter values} $\p^{(i)}$ for $i=1,\ldots,m$. 

To use $H_r(s,\p)$ for an unsampled parameter value {$\hat{\p}$} and to still guarantee bounded error, 
we compute $\hat{\v} = \hat{\P}^{-1}\mathbf{1}$, construct the new basis $\widehat{V} =  \begin{bmatrix}
          V & \hat{\v}
    \end{bmatrix}$,
and obtain $H_r(s,\hat{\p})$ as 
in (\ref{parametric reduced matrices}), now using $\widehat{V}$. 
Theorem \ref{Residue} will then guarantee a bounded error at 
$\hat{\p}$ as well. 

The reduction step (\ref{parametric reduced matrices})  does
not need to be applied from scratch for every new $\hat{\p}$. For the new basis $\widehat{V}$, consider $\widehat{M}_r$ : $ \widehat{M}_r = \widehat{V}^T M \widehat{V}
     = 
     \begin{bmatrix}
          V^T M V & V^{T} M\hat{\v} \\ \hat{\v}{^{T}}MV & \hat{\v}{^{T}}M \hat{\v}
     \end{bmatrix}.$
The terms  
$V^T M V$, $V^{T}M$ and $MV$ are calculated only once in the offline stage using $V$, and only the vector $M \hat{v}$ needs computing for a  new parameter
$\hat \p$. The situation is similar for the other reduced quantities except for $\widehat{L}_r(\p)$ due to the nonaffine parametrization of $L(\p) = \P L \P$. An affine parametric approximation of $L(\p)$ to allow efficient online computations, via DEIM, for example,  \cite{SerkanSurvey}, will be studied in a future work.

\subsection{Smaller number of parameters}
Now we assume that $L(\p)$ is parametrized with a smaller number of parameters.  Let $\p = [\p_1~\p_2~\cdots~\p_\nu]^T \in \Omega_\nu \subseteq \IR^\nu$ and consider the parametrization
\begin{align} \label{nuparam}
L(\p) = \P L \P~~\mbox{with}~~\P = \mathsf{diag}(\p_1 I_{n_1}, \ldots, \p_\nu I_{n_\nu}),
\end{align}
where $n_1+ \cdots + n_\nu = n$
and $\nu < n$. This can be viewed as some of the peak voltage magnitudes $E_i$  varying together. This structure will drastically simplify the algorithmic considerations from
Section \ref{sec:alg}. 
In (\ref{nuparam}) we can also set some $p_i$'s to $1$ to allow variations only in a subset set $E_i$'s.
{
\begin{proposition} \label{prop:smallparam}
Consider the parametrization in (\ref{nuparam}). Let $\mathbf{0}_q \in \IR^q$ denote the zero vector and define
\begin{equation}
e_k = \begin{bmatrix}
\mathbf{0}^T_{n_1+\cdots+n_{k-1}} & \mathbf{1}^T_{n_k} & \mathbf{0}^T_{n_{k+1}+\cdots+n_\nu}
\end{bmatrix}^T \in \IR^n
\end{equation}
 for $k=1,2,\ldots,\nu$. If 
$\{e_1,e_2,\ldots,e_{\nu}\} \in \mathsf{span}(V)$, then 
$H_r(s,\p)$ retains the simple pole at zero and its corresponding parameter-dependent residue of $H(s,\p)$ for every $\p \in \Omega_\nu$. 
\end{proposition}
\begin{proof}
For any  $\hat{\p} \in \Omega_\nu$, 
${\hat{\v}} = \begin{bmatrix}
      \frac{1}{\p_1} \mathbf{1_{n_1}} & \cdots & \frac{1}{\p_{\nu}} \mathbf{1_{n_\nu}}
\end{bmatrix}^T$ is 
the eigenvector of $L(\hat{\p})$
corresponding to the zero eigenvalue. Note that
$ {\hat{\v}}=\frac{1}{\p_1}e_1 + \cdots + \frac{1}{\p_\nu}e_\nu
$. Therefore, if $\{e_1,e_2,\ldots,e_{\nu}\} \in \mathsf{span}(V)$, we have 
$\hat{v}\in \mathsf{span}(V)$ for \emph{every} $\hat{\p} \in \Omega_\nu$ and the desired result follows from
Theorem \ref{Residue}.
\end{proof}
}
Proposition \ref{prop:smallparam} reveals that in the case of the parametrization  (\ref{nuparam}), adding $\nu$ vectors to the span of $V$ will be enough to match the residue at $s=0$ for \emph{every} $\p \in \Omega_\nu$. Therefore, augmenting the global basis by a new vector for a given $\hat{\p}$ as explained in Section \ref{sec:alg} is no longer necessary. A \emph{fixed}  global  basis $V$ satisfying $\{e_1,e_2,\ldots,e_{\nu}\} \in \mathsf{span}(V)$ does the job for every $\p \in \Omega_\nu$. Note that one needs $\nu$ to be modest so that the reduced dimension stays modest. 

\subsubsection{Algorithmic details for implementing Proposition \ref{prop:smallparam}} \label{sec:propalg}
The global basis $V$ in Proposition \ref{prop:smallparam} can result from any model reduction method of choice. As long as the vectors
$\{e_1,\ldots,e_\nu\}$ are added to its span, the result will hold. We will form $V$ 
as in (\ref{parametric reduction bases})
where the local bases  result from the modified implementation of 
SOR-IRKA as described in Section \ref{sec:alg}.  
Given the parameter samples $\p^{(i)}$ for $i=1,\ldots,m$, let $\v^{(i)}$ 
denote the eigenvector of 
$L(\p^{(i)})$ corresponding to the zero eigenvalue. Our SOR-IRKA implementation will provide that $\{\v^{(1)},\dots,\v^{(m)}\} \in \mathsf{span}(V)$. As shown in the proof of Proposition \ref{prop:smallparam},
for any $\hat\p\in \Omega_\nu$, $\hat \v = \hat\P^{-1} \mathbf{1}$ is spanned by $\nu$ vectors. We will choose $m \geq \nu$ different parameter samples, obtaining a linearly independent set 
$\{\v^{(1)},\dots,\v^{(m)}\}$. Since these vectors are 
in the span of $V$, we will automatically satisfy the subspace condition in Proposition \ref{prop:smallparam}. Therefore, our construction of $V$ via modified SOR-IRKA with $m \geq \nu$ parameter samples will  guarantee bounded $\mathcal{H}_2$ and $\mathcal{H}_\infty$ error for every $\p \in \Omega_\nu$ without explicitly adding the vectors $\{e_1,\ldots,e_\nu\}$ to the model reduction basis $V$.

\section{Numerical results} \label{sec:num}
We use a linearized model of $2736$-bus Polish network\cite{zimmerman2016matpower} with  $n = 2736$. We focus on a single-input single-output model with 
$B = C^T = [1~~0~~\cdots~~0]^T \in {\rm I\!R}^{n \times 1} $
and allow $15\%$ variation in  peak voltage magnitudes, i.e., $0.85 \leq \p_i \leq 1.15$ in $L(\p)$. Recall that $p_i=1$ corresponds to the non-parametric unity voltage magnitude case ($E_i=1$). 

\subsection{Case 1: two parameters} \label{ex2}
We consider a parametrization with $\nu=2$ parameters $\p_1$ and $\p_2$ as
$ \P = \text{diag}(\p_1I_{\frac{n}{2}},\p_2I_{\frac{n}{2}})$. We pick two random samples, namely $\p^{(1)} = [0.9572~0.93399]^T$ and $\p^{(2)} = [1.0304~0.9522]^T$, and apply the modified SOR-IRKA to obtain local bases
${V^{(1)}} \in \IR^{n\times 20}$ and ${V^{(2)}}\in \IR^{n\times 20}$. An orthogonalization of $[V^{(1)}~V^{(2)}]$ leads to the global basis $V \in \IR^{n \times 40}$, thus a reduced model $H_r(s,\p)$ with $r=40$. Due to Proposition \ref{prop:smallparam} and the discussion in
Section \ref{sec:propalg}, $H_r(s,\p)$ matches the residue at $s=0$ and provides bounded   $\mathcal{H}_2$ and  $\mathcal{H}_\infty$ error throughout the whole domain  $[\p_1,~\p_2] \in \Omega_2 = [0.85 \ 1.15] \times [0.85 \ 1.15]$. 
{To illustrate the accuracy of $H_r(s,\p)$,   in Figure \ref{fig:Serkan_Case2_npis20_Nis2500} we show the  relative $\mathcal{H}_\infty$ error over the full parameter space.}  As  the figure illustrates, the structure-preserving reduced model $H_r(s,\p)$ is a high fidelity approximation to $H(s,\p)$ over the full parameter space  with a maximum relative error less than {$1.5 \times 10^{-2}$}. 
 \begin{figure}
  \includegraphics[width=0.75\linewidth]{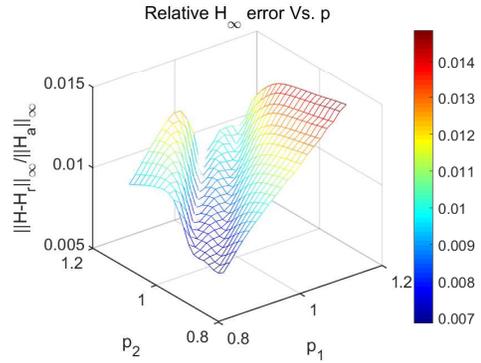}
  \centering
  \caption{Example \ref{ex2}: Relative $\mathcal{H}_\infty$ error over the parameter domain }
  \label{fig:Serkan_Case2_npis20_Nis2500}
\end{figure}
\subsection{Case 2: four parameters} \label{ex3}
In this example, we consider paremetrization via four parameters $\p_1$, $\p_2$, $\p_3$ and $\p_4$ to generate the matrix $\P$ such that  $\P = \text{diag}(\p_1I_{\frac{n}{4}},\p_2I_{\frac{n}{4}},\p_3I_{\frac{n}{4}},\p_4I_{\frac{n}{4}})$.
 We randomly pick  four parameter sample sets: 
 \begin{center}
 \begin{tabular}{|c|| c c c c |} 
 \hline
  Sample set & $\p_1$ & $\p_2$ & $\p_3$ & $\p_4$ \\
  \hline
 $\p^{(1)}$ & 1.0967& 0.8541  & 0.9399  &  0.887  \\
 $\p^{(2)}$ & 0.9399  &  0.9146 &  1.0377  & 1.0459 \\
 $\p^{(3)}$ & 0.9522  & 1.0713 & 0.9399   & 0.9572  \\
 $\p^{(4)}$ & 1.0801  & 0.9399 &  1.0377 & 1.1029 \\
  \hline
 \end{tabular}
 \end{center}
  Then using these samples, we apply the modified SOR-IRKA to obtain the local bases $V^{(i)}\in \IR^{n\times 20};~ i=\{1,2,3,4\}$ and  a parametric reduced model of order $r = 80 $  ($V\in \IR^{n\times 80})$. As in the previous example, 
 this reduced model guarantees bounded error over the whole parameter space. 
  To show the approximation quality,
 we  pick $200$ random samples in the four-dimensional parameter space, and depict the resulting relative $\mathcal{H}_\infty$ error  in Figure \ref{fig:Case3_RelError_r80_200samples},
 showing a maximum relative error less than $10^{-2}$ over this sample set. 
 \begin{figure}
  \includegraphics[width=0.75\linewidth]{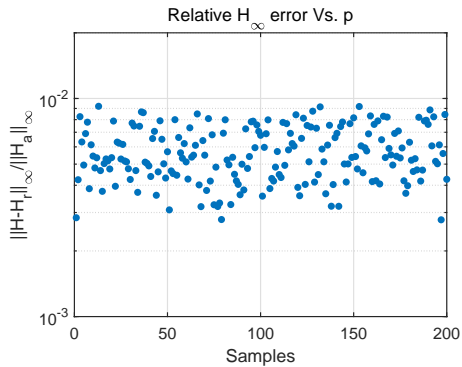}
  \centering
  \caption{Example \ref{ex3}: Relative $\mathcal{H}_\infty$ error over 200 samples.}
  \label{fig:Case3_RelError_r80_200samples}
\end{figure}

\section{Conclusions and Future Work}  \label{sec:conc}
We have developed a structure-preserving parametric model reduction approach for linearized swing equations using a global basis approach and $\mathcal{H}_2$-based interpolatory model reduction. We 
have established the subspace conditions for the model reduction basis so that the error system is an $\mathcal{H}_2$ and $\mathcal{H}_\infty$ function over the entire parameter space. The efficiency of our proposed approach has been illustrated via two numerical examples.

Parameter sampling for constructing the local bases was not the focus of this work. Any efficient parameter selection methodology can be incorporated into our framework and will be considered in a future together with the recent composite $\mathcal{H}_2\times \mathcal{L}_2$-optimal basis constructions
\cite{morHunMS18,grimm2018parametric}. Extensions to the nonlinear parametric setting is also an important topic to consider. \\[-3ex]

\section*{Acknowledgements} 
{We thank Dr.  Vassilis Kekatos and Dr. Siddharth Bhela
 for various discussions and for providing the 2736-bus Polish network model.} 

\bibliographystyle{plain}
\bibliography{ACC2020}

\end{document}